\documentclass{article}
\usepackage[utf8]{inputenc}
\usepackage{authblk}
\usepackage{setspace}
\usepackage[margin=1.25in]{geometry}
\usepackage{graphicx}
\graphicspath{ {./figures/} }
\usepackage{subcaption}
\usepackage{amsmath}
\usepackage{amsthm}
\usepackage{lineno}
\usepackage[english]{babel}
\usepackage{hyperref}
\usepackage{tikz}
\usetikzlibrary{arrows,positioning,shapes.geometric,arrows.meta}
\usepackage{varwidth}
\usepackage{pgfplots}
\pgfplotsset{compat = newest}
\usepgfplotslibrary{fillbetween}
\usepackage{algorithm}
\usepackage[noend]{algpseudocode}
\usepackage{booktabs}
\newcommand\mr[1]{\multicolumn{2}{r}{#1}}

\usepackage[style=ieee, 
citestyle=numeric-comp,
sorting=none]{biblatex}
\newtheorem{theorem}{Theorem}
\newtheorem{example}{Example}

\newlength{\commentlen}

\title{Increasing Ticketing Allocative Efficiency Using Marginal Price Auction Theory}

\author[]{Boxiang Fu}

\affil[]{boxiangf@student.unimelb.edu.au \\ 
Faculty of Science, University of Melbourne, Melbourne, VIC 3010, Australia}

\date{\today}

\begin{document}

\maketitle

\begin{abstract}
Most modern ticketing systems rely on a first-come-first-serve or randomized allocation system to determine the allocation of tickets. Such systems has received considerable backlash in recent years due to its inequitable allotment and allocative inefficiency. We analyze a ticketing protocol based on a variation of the marginal price auction system. Users submit bids to the protocol based on their own utilities. The protocol awards tickets to the highest bidders and determines the final ticket price paid by all bidders using the lowest winning submitted bid. Game theoretic proof is provided to ensure the protocol more efficiently allocates the tickets to the bidders with the highest utilities. We also prove that the protocol extracts more economic rents for the event organizers and the non-optimality of ticket scalping under time-invariant bidder utilities.
\end{abstract}


\section{Introduction}

Current ticket allocation systems used by most major ticketing websites operate on a first-come-first-serve or randomized allocation basis. Such a system has caused considerable backlash over recent years due to its opaque criteria for allocation and the need to compete for who can refresh the ticketing webpage the fastest in the milliseconds after tickets are released for sale (see Ref. \cite{Maiden2023}). Economically, current systems are also largely inefficient in allocating the tickets to the consumers with the highest utility for the tickets, thereby resulting in a loss in total allocative efficiency.

We propose a ticketing protocol based on the marginal price auction system. The protocol allocates the tickets to the bidders with the highest bids and the price paid by all bidders is the lowest winning submitted bid. The protocol provably increases the total allocative efficiency compared to current allocation systems by assigning the tickets to the group of consumers with the highest utility. We also prove that the proposed system increases the economic rents extracted for the seller as well as offering a partial solution to the ticket scalping problem by proving that rational bidders with time-invariant utilities will refrain from buying scalped tickets.

\section{Protocol Description}

We begin by briefly summarizing ticketing systems based on a first-come-first-serve protocol (see Ref. \cite{Anthony2016}). Prior to the tickets going on sale, the seller publicly announces a time at which the bulk of the tickets are available for purchase. Users typically enter into the ticketing webpage prior to the tickets going on sale and compete on refreshing the webpage immediately after the ticket sale time commences. Users are then served based on their chronological time-stamp registered with the ticketing webpage. The tickets are progressively sold until the allotment has been exhausted or until all users wishing to purchase has been served. Fig. \ref{FigFCFS} briefly outlines the timeline of a first-come-first-serve ticketing system.

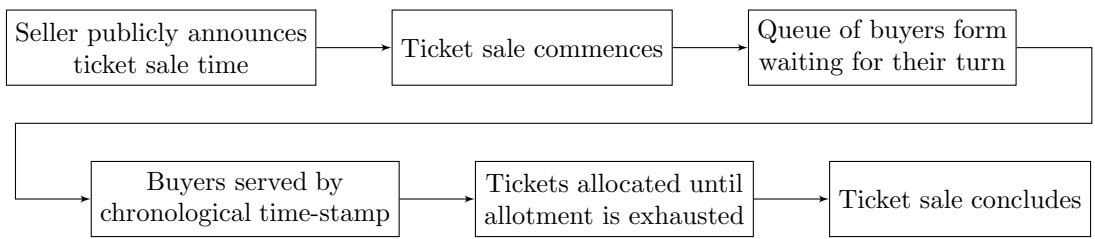
\begin{figure}[h]
    \centering

    \begin{tikzpicture}[>=latex']
        \tikzset{block/.style= {draw, rectangle, align=center,minimum width=2cm,minimum height=1cm},
        rblock/.style={draw, shape=rectangle,rounded corners=1.5em,align=center,minimum width=2cm,minimum height=1cm},
        input/.style={ 
        draw,
        trapezium,
        trapezium left angle=60,
        trapezium right angle=120,
        minimum width=2cm,
        align=center,
        minimum height=1cm
    },
        }
        \node [block]  (start) {Seller publicly announces \\ ticket sale time};
        \node [block, right =1cm of start] (node1) {Ticket sale commences};
        \node [block, right =1cm of node1] (node3) {Queue of buyers form \\ waiting for their turn};
        \node [block, below right =1cm and -3cm of start] (node4) {Buyers served by \\ chronological time-stamp};
        \node [block, right =1cm of node4] (node5) {Tickets allocated until \\ allotment is exhausted};
        \node [block, right =1cm of node5] (node6) {Ticket sale concludes};
        \node [coordinate, below right =0.5cm and 1cm of node3] (right) {};  
        \node [coordinate,above left =0.5cm and 1cm of node4] (left) {};  
        \path[draw,->] (start) edge (node1)
                    (node1) edge (node3)
                    (node3.east) -| (right) -- (left) |- (node4)
                    (node4) edge (node5)
                    (node5) edge (node6)
                    ;
    \end{tikzpicture}

    \caption{Timeline of Key Steps in a First-Come-First-Serve Ticketing System}
    \label{FigFCFS}
\end{figure}

Such a system is inefficient both in terms of time and allocation. Most first-come-first-serve systems require the user to be physically on the webpage waiting in the queue to be able to participate in the allocation, with queuing time possibly taking hours for large events (see Ref. \cite{Maiden2023} for the case of Taylor Swift’s 2023 Australian tour). Economically, the system is also not allocative efficient in most cases. In the common case where demand exceeds supply, the first-come-first-serve system allocates tickets based on chronological ordering, and potentially leaves many buyers with higher utility without an allocation (see Fig. \ref{FigAllocEff} and Example \ref{Example1}).

We propose an alternative system for ticket allocation based on the marginal price auction system. The system is a multi-unit generalization of the Vickrey auction system (see Ref. \cite{Vickrey1961}). In a marginal price auction, a fixed number of units of a homogeneous commodity is put forward for auction. Bidders submit bids for the units via a (usually) sealed-bid auction. The auctioneer allocates the units to the bidders with the highest bids until the allocation is exhausted. The price paid on each unit for all bidders is the lowest winning submitted bid (see Fig. \ref{FigVickrey}). The marginal price auction system has some particularly useful game theoretic properties that are explored in the next section. For now, we outline our proposed ticket allocation mechanism.

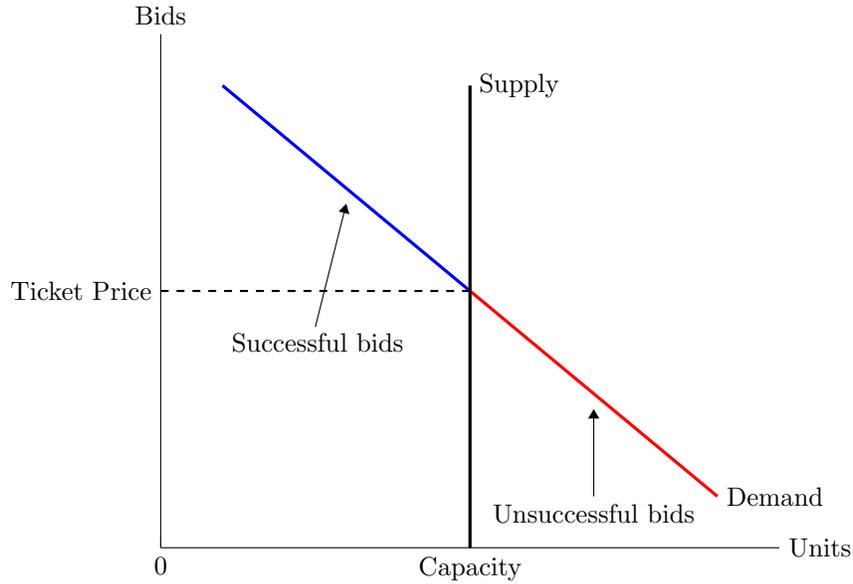
\begin{figure}[h]

    \begin{center}
    \begin{tikzpicture}
    \begin{axis}[
    scale = 1.2,
    xmin = 0, xmax = 10,
    ymin = 0, ymax = 10,
    axis lines* = left,
    xtick = {0}, ytick = \empty,
    clip = false,
    ]
    \addplot[color = blue, very thick] coordinates {(1, 9) (5, 5)};
    \addplot[color = red, very thick] coordinates {(5, 5) (9, 1)};
    \addplot[color = black, very thick] coordinates {(5, 0) (5, 9)};
    
    \addplot[color = black, dashed, thick] coordinates {(0, 5) (5, 5)};
    
    \node [right] at (current axis.right of origin) {Units};
    \node [above] at (current axis.above origin) {Bids};
    \node [above] at (5, 5.2) {};
    \node [left] at (0, 5) {Ticket Price};
    \node [below] at (5, 0) {Capacity};
    \node [right] at (9, 1) {Demand};
    \node [right] at (5, 9) {Supply};
    
    \node [right] at (1, 4) {Successful bids};
    \node [above] at (7, 0.3) {Unsuccessful bids};
    
    \draw[-Triangle] (2.5, 4.3) to (3, 6.7);
    \draw[-Triangle] (7, 1) to (7, 2.7);
    
    \end{axis}
    \end{tikzpicture}
    \end{center}
    
    \caption{Ticket Allocation and Pricing in a Marginal Price Auction System}
    \label{FigVickrey}
\end{figure}

The timeline of our proposed marginal price ticket allocation system is outlined in Fig. \ref{FigMPS}. Instead of publicly announcing a ticket sale commencement time, the seller instead announces a time window for bid submission. During this window, bidders are free to submit bids for one or more tickets. Collateral may be taken to ensure the bid is genuine. A price floor may also be optionally implemented by the seller so that only bids exceeding the floor are accepted. Once the time window elapses, bidding is closed and all outstanding bids are entered into the auction. A marginal price auction system ranks the bids according to their monetary amount and allocates tickets to the highest bids until the allocation is exhausted. The price paid is determined by the lowest winning submitted bid. Tickets are then released to the successful bidders with a requirement to pay the ticket price within a set timeframe and any excess collateral or rebates is released back to the bidders.

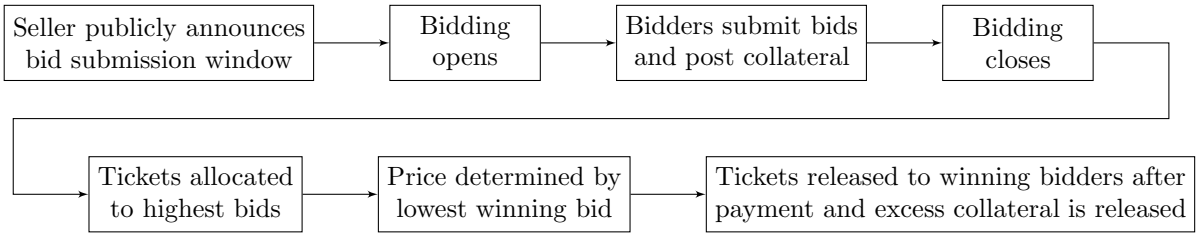
\begin{figure}[h]
    \centering

    \begin{tikzpicture}[>=latex']
        \tikzset{block/.style= {draw, rectangle, align=center,minimum width=2cm,minimum height=1cm},
        rblock/.style={draw, shape=rectangle,rounded corners=1.5em,align=center,minimum width=2cm,minimum height=1cm},
        input/.style={ 
        draw,
        trapezium,
        trapezium left angle=60,
        trapezium right angle=120,
        minimum width=2cm,
        align=center,
        minimum height=1cm
    },
        }
        \node [block]  (start) {Seller publicly announces \\ bid submission window};
        \node [block, right =1cm of start] (node1) {Bidding \\ opens};
        \node [block, right =1cm of node1] (node2) {Bidders submit bids \\ and post collateral};
        \node [block, right =1cm of node2] (node3) {Bidding \\ closes};
        \node [block, below right =1cm and -3cm of start] (node4) {Tickets allocated \\ to highest bids};
        \node [block, right =1cm of node4] (node5) {Price determined by \\ lowest winning bid};
        \node [block, right =1cm of node5] (node6) {Tickets released to winning bidders after \\payment and excess collateral is released};
        \node [coordinate, below right =0.5cm and 1cm of node3] (right) {};  
        \node [coordinate,above left =0.5cm and 1cm of node4] (left) {};  
        \path[draw,->] (start) edge (node1)
                    (node1) edge (node2)
                    (node2) edge (node3)
                    (node3.east) -| (right) -- (left) |- (node4)
                    (node4) edge (node5)
                    (node5) edge (node6)
                    ;
    \end{tikzpicture}

    \caption{Timeline of Key Steps in a Marginal Price Ticketing System}
    \label{FigMPS}
\end{figure}

The protocol for the marginal price allocation mechanism is summarized in Fig. \ref{FigMechanism}. After the bidding window is opened, users are first required to validate their identities if they have not done so prior. This entails signing up to the protocol so that an unique identifier can be attributed to the user (see Ref. \cite{McMurry2017}). For users wishing to bid multiple units, multiple identifiers should be provided by the user. These should ideally be the identities of the individuals hoping to attend the event. Such identification is crucial to allow us to associate a user submitting multiple bids as a proxy for multiple natural persons submitting multiple one unit bids. This allows us to ensure the validity of Theorem \ref{Theorem1} and also reduce potential malicious activity such as intentionally bidding in large quantities by ticket scalpers to reduce overall available supply.

\begin{figure}[h]
    \centering

    \begin{tikzpicture}[>=latex']
        \tikzset{block/.style= {draw, rectangle, align=center,minimum width=6.5cm,minimum height=1.5cm},
        rblock/.style={draw, shape=rectangle,rounded corners=1.5em,align=center,minimum width=5cm,minimum height=1cm},
        input/.style={ 
        draw,
        trapezium,
        trapezium left angle=60,
        trapezium right angle=120,
        minimum width=2cm,
        align=center,
        minimum height=1cm
    },
        }
        \node [block]  (start) {1. Bidding window opens and \\ price floor is announced};
        \node [block, below =1cm of start] (node1) {2. Protocol validates identity \\ of potential bidders};
        \node [block, below =1cm of node1] (node2) {3. Protocol receives bids \\ and collateral};
        \node [block, below =1cm of node2] (node3) {4. OPTIONAL:\\ Disclose indicative final price \\ to stimulate bidding};
        \node [block, below =1cm of node3] (node4) {5. Bidding window closes};
        \node [block, right =1cm of node4] (node5) {6. Bids ranked according to descending \\ price order};
        \node [block, above =1cm of node5] (node6) {7. Tickets allocated to highest \\ bids until exhausted or price \\ floor is reached};
        \node [block, above =1cm of node6] (node7) {8. Price of tickets determined \\ by price of lowest winning bid};
        \node [block, above =1cm of node7] (node8) {9. Tickets released to successful \\ bidders after payment and excess \\ collateral is released};
        \node [block, above =1cm of node8] (node9) {10. OPTIONAL:\\ Rebates may be distributed if settlement \\ price greatly exceeds floor price};

        \node [coordinate, below right =0.5cm and 1cm of node3] (right) {};  
        \node [coordinate,above left =0.5cm and 1cm of node4] (left) {};  
        \path[draw,->] (start) edge (node1)
                    (node1) edge (node2)
                    (node2) edge (node3)
                    (node3) edge (node4)
                    (node4) edge (node5)
                    (node5) edge (node6)
                    (node6) edge (node7)
                    (node7) edge (node8)
                    (node8) edge (node9)
                    ;
    \end{tikzpicture}

    \caption{Description of Steps in a Marginal Price Ticketing Protocol}
    \label{FigMechanism}
\end{figure}
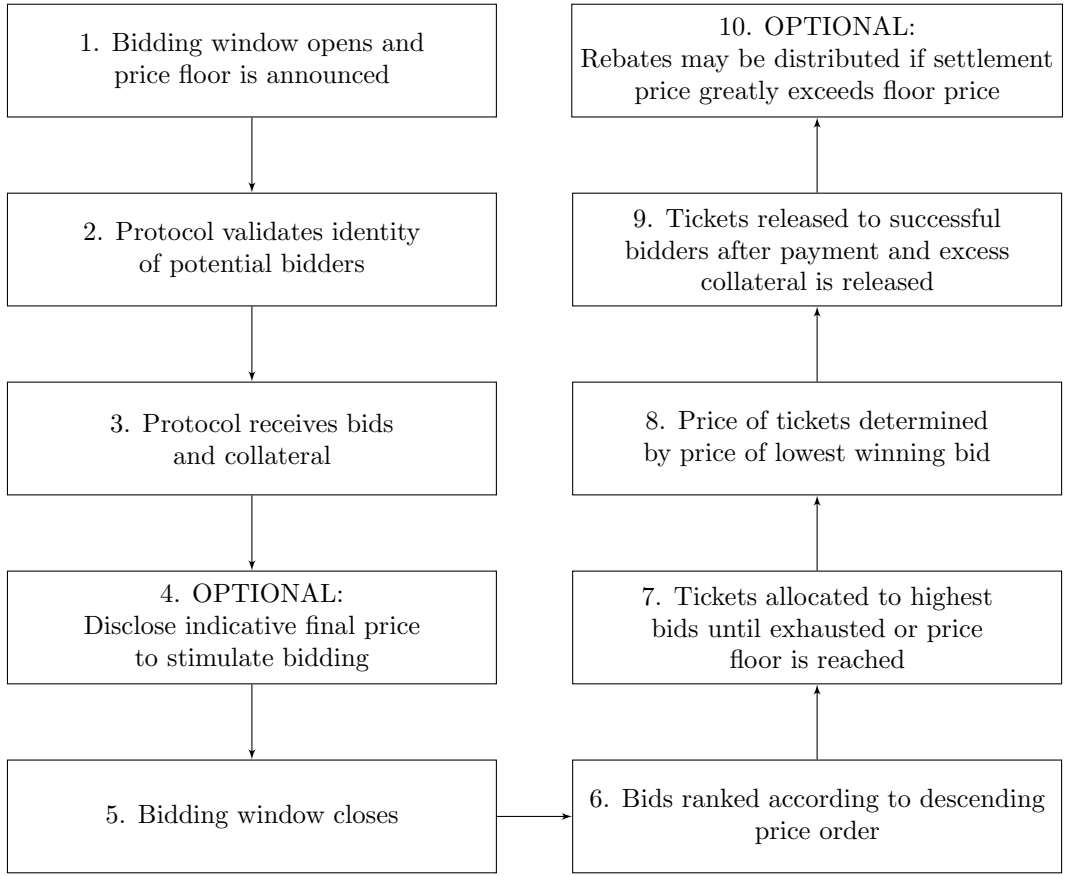

Once user identification is validated, bids may be submitted through the protocol and bids exceeding the price floor are entered into the central database. Ideally, collateral equalling 100\% of the bid amount should also be posted concurrently with the bid to ensure the bid is genuine. This may be relaxed to cover less than 100\% if additional guarantees can be put in place to ensure the bid is honest (e.g. the number of times the user has bid, the number of verified identities associated with the user, etc). This step can also provide useful information to the event organizers to gauge the popularity of the event. If the number of submitted bids greatly exceeds capacity, it could allow organizers to schedule additional shows to increase supply.

Next, the event organizers may optionally choose to disclose an indicative final price prior to the end of the bidding window to stimulate bidding. This could be as rudimentary as determining the lowest winning bid of all the submitted bids up until this time. However, since the auction is no longer sealed-bid, its dynamics may be affected and the optimal bidding strategy may not be the one proven in Theorem \ref{Theorem1}.

Once the bidding window elapses, the bidding webpage closes and the protocol no longer accepts incoming bids. The protocol then initiates a marginal price auction on all outstanding bids (see Algorithm \ref{CodeAuction}). Bids are ranked in descending price order and tickets are allocated to the highest bids until the ticket allocation is exhausted, and the price of all tickets is determined by the lowest winning bid. In the case of multiple bids at the lowest winning bid price, a randomized lottery or the chronological order of the bids may be used to allocate the remaining tickets.

After the auction is executed, the tickets are released to the successful bidders and any excess collateral is released. If the collateral amount is less than the final ticket price, the bidder may be required to pay the remaining amount within a predetermined settlement period. Optionally, a rebate (both monetary and/or non-monetary) could be distributed to the winning bidders after the auction should the final settlement price greatly exceed the original price floor ticket price. Its rationale is explained in the next section.

\begin{algorithm}
\setlength{\commentlen}{38ex}
\caption{Marginal Price Ticket Auction}\label{CodeAuction}
\begin{algorithmic}
\Require $\textbf{b}=\left(b_1, b_2, \ldots  b_i, \ldots b_N\right)$ \Comment{\makebox[\commentlen][l]{Submitted bids in chronological order}}
\Require $m$ \Comment{\makebox[\commentlen][l]{Price floor}}
\Require $K$ \Comment{\makebox[\commentlen][l]{Number of available tickets}}

\If{$N \leq K$}
    \State \Return User identifiers of $\textbf{b}$ and ticket price $m$
\ElsIf{$N > K$}
    \State $\textbf{c} \gets$ DescendingSort(\textbf{b})
    \State \Return User identifiers of $c_i$ with $i \leq K$ and ticket price $c_K$
\EndIf
\end{algorithmic}
\end{algorithm}

\section{Properties and Proofs}

A marginal price auction system has a number of nice game theoretic properties that allows the system to more efficiently allocate tickets based on the user's individual valuations. In essence, the marginal price auction system allocates the tickets to the group with the highest utility for the tickets, as opposed to a first-come-first-serve allocation in conventional ticketing systems. First, we prove that for rational bidders with demand for only one ticket, the optimal strategy for each bidder is to bid their true value of the item. From this, we show that the marginal price auction system extracts economic rents for the seller that is greater than or equal to the rents extracted from the first-come-first-serve system. We also show that the total valuation of successful bidders from the marginal price auction system is greater than or equal to the total valuation of the successful bidders from the first-come-first-serve system. This increases allocative efficiency and allots the limited number of tickets available to the group of bidders with the highest valuations. Finally, we show that the system offers a partial solution to the ticket scalping problem by proving that it is not optimal to buy from scalpers in the case of time-invariant bidder valuations.

The first theorem is a standard result of marginal price auction systems found in most auction theory textbooks. The exposition used here is based on a variation of the proof found in Ref. \cite{Krishna2003}. Throughout this section we assume that each bidder has demand for one ticket only. This is a valid assumption in the case of event ticketing problems as one person can only maximally enjoy one unit of the ticket by being physically present at the event. We relax this one ticket assumption in the protocol implementation description by introducing an identity verification mechanism so that a user submitting multiple bids can be regarded as a proxy for multiple natural persons submitting multiple one unit bids. For ease of exposition we regard users that bid at exactly the final price as losing the bid (i.e. they are left without a ticket). For physical implementation purposes, a randomization procedure may be used so that all bidders who bid at exactly the final price is entered into a lottery and a subset is randomly chosen to be allocated the remaining tickets.

\begin{theorem}
\label{Theorem1}
In a marginal price auction with single-unit bidder demand, the optimal strategy for all bidders is to bid their own true valuation.
\end{theorem}

\begin{proof}

Let $N$ denote the number of bidders in the auction and $K$ denote the number of available units with $N > K$. Also, let $v_i$ denote bidder $i$'s valuation for one unit of the item, $b_i$ denote bidder $i$'s submitted single-unit bid for the item, and let $\textbf{c}=\left(c_1, c_2, \ldots  c_i, \ldots c_N\right)$ denote the $N$-vector of submitted bids by the $N$ bidders arranged in descending price order (similarly $\textbf{c}^{-i}$ is the descending order bid vector without bid $i$).

The final price set by the marginal price auction is given by the lowest winning bid at
\begin{equation*}
    p=c_{K}
\end{equation*}

The payoff to bidder $i$ is given by the payoff function
\begin{equation*}
    P_i(v_i, \textbf{c}) =
        \begin{cases}
          v_i-p & \text{if } b_i > p\\
          0 & \text{otherwise}
        \end{cases}   
\end{equation*}

We claim that $b_i=v_i$. Suppose by contradiction that $b_i>v_i$. We have the following cases:

\textit{Case 1}: $p \geq b_i>v_i$. Bidder $i$ loses the auction and receives a payoff of 0 regardless of their action.

\textit{Case 2}: $b_i > p \geq v_i$. The payoff to bidder $i$ is $P_i(v_i, \textbf{c})=v_i-p \leq 0$ and is weakly dominated by the alternate strategy $\Tilde{b_i}=v_i$ with payoff $P_i(v_i, \textbf{c}^{-i}, \Tilde{b_i})=0$.

\textit{Case 3}: $b_i > v_i > p$. Since both $b_i > p = c_{K}$ and $v_i > p = c_{K}$, it makes no difference bidding at $b_i$ or $v_i$ as it only permutes the location of bidder $i$'s bid in the first ${K-1}$ places of vector $\textbf{c}$. So bidder $i$ wins the bid regardless and pays the same price $p=c_{K}$.

The three exhaustive cases shows that the strategy $b_i>v_i$ is weakly dominated by the strategy $\Tilde{b_i}=v_i$. Next, suppose that $b_i<v_i$. We have the following cases:

\textit{Case 1}: $p \geq v_i > b_i$. Bidder $i$ loses the auction and receives a payoff of 0 regardless of their action.

\textit{Case 2}: $v_i > p \geq b_i$. The payoff to bidder $i$ is $P_i(v_i, \textbf{c})= 0$ and is weakly dominated by the alternate strategy $\Tilde{b_i}=v_i$ with payoff 

\begin{equation*}
    P_i(v_i, \textbf{c}^{-i}, \Tilde{b_i}) =
        \begin{cases}
          v_i-\Tilde{p} & \text{if } v_i > c_{K-1}\\
          0 & \text{otherwise}
        \end{cases}   
\end{equation*}
where $\Tilde{p}=c_{K-1}$ is now the lowest winning bid due to the insertion of bid $\Tilde{b_i}$ into the first $K-1$ slots of $\textbf{c}$.

\textit{Case 3}: $v_i > b_i > p$. As with the previous \textit{Case 3}, bidder $i$ wins the bid regardless and pays the same price $p=c_{K}$.

Thus, both strategies $b_i>v_i$ and $b_i<v_i$ are weakly dominated by $\Tilde{b_i}=v_i$. We conclude that the optimal bidding strategy for bidder $i$ is to bid their own true valuation.

\end{proof}

The theorem above is not true in general if bidders have demand for more than one unit (see Ref. \cite{Krishna2003}). Hence, an identity verification mechanism is needed so that we regard a user submitting multiple bids as proxies for multiple natural persons. The mechanism effectively allows the seller to circumvent determining the pricing of the tickets based on imperfect information and instead rely on the marginal price auction mechanism to allow bidders to reveal their own reservation price through the bidding process. The theorem above guarantees that rational bidders will reveal their own willingness-to-pay during the bidding process and disclose this information to the seller. The mechanism also allows the seller to extract more economic rents than the first-come-first-serve system, which we will prove below. We also impose a price floor at which bids must exceed to be successful at being allocated a ticket. This is typical in most modern ticketing systems (it is just the ticket price in first-come-first-serve systems).

\begin{theorem}
\label{Theorem2}
In a marginal price auction with single-unit bidder demand and price floor, the economic rents extracted is greater than or equal to the economic rents extracted from a first-come-first-serve system.
\end{theorem}

\begin{proof}

    Let $\textbf{c}=\left(c_1, c_2, \ldots  c_i, \ldots c_N\right)$ denote the $N$-vector of submitted bids by the $N$ bidders arranged in descending price order. Let the price floor be denoted by \$$m$ and $K$ denote the number of units available to bid with $N > K$. We have the following cases:

    \textit{Case 1}: $c_{K} \geq m$. There is enough demand above the price floor to exhaust the supply of $K$ units available to bid. The economic rents obtained by the first-come-first-serve system is given by $mK$ (allocated to the first $K$ bidders with bids exceeding the price floor in chronological order), while the the economic rents obtained by the marginal price auction is given by $c_{K}K$. Since $c_{K} \geq m$, we have $c_{K}K \geq mK$.

    \textit{Case 2}: $c_{K} < m$. There is not enough demand above the price floor to exhaust the supply of $K$ units available to bid. The price floor ensures that only the $k<K$ bidders with $c_1, c_2, \ldots c_k \geq m$ are allocated at price \$$m$ and $K-k$ units are left unallocated. The economic rents extracted is $mk$ for both systems.

    From the two cases, we conclude that the a marginal price auction extracts economic rents that is greater than or equal to that extracted from a first-come-first-serve system.
    
\end{proof}

Below we provide a two simple examples on the different economic rents extracted by both systems.

\begin{example}
    \label{Example1}
    \normalfont
    Let the number of bidders be $N=6$ and the number of units available to bid be $K=3$. Let the price floor be $m=20$ with the chronological bid vector $\textbf{b}=\left(35, 15, 40, 20, 25, 20 \right)$. The descending price vector is then $\textbf{c}=\left(40, 35, 25, 20, 20, 15 \right)$.

    The first-come-first-serve system sets the ticket price at $m=20$ and the successful bidders are the 1st, 3rd, and 4th entries in the chronological bid vector $\textbf{b}$. The economic rents extracted for the seller is $3 \times 20 = 60$.

    The marginal price auction system sets the ticket price at $c_3 = 25$ and the successful bidders entered into the auction in the chronological order of 1st, 3rd, and 5th. The economic rents extracted for the seller is $3 \times 25 = 75$. The excess economic rents extracted amounts to \$$15$ and the 4th chronologically-ordered bidder would no longer be successful in the auction.
    
\end{example}

\begin{example}
    \normalfont
    Let the number of bidders be $N=6$ and the number of units available to bid be $K=3$. Let the price floor be $m=30$ with the chronological bid vector $\textbf{b}=\left(35, 15, 40, 20, 25, 20 \right)$. The descending price vector is then $\textbf{c}=\left(40, 35, 25, 20, 20, 15 \right)$.

    Both systems set the ticket price at the price floor $m=30$ and the successful bidders are the 1st and 3rd entries in the chronological bid vector $\textbf{b}$. The economic rents extracted for both systems is $2 \times 30 = 60$. In this scenario, the seller may consider lowering the price floor prior to the bidding window closing to allow enough bids to exceed the price floor so that all units are allocated.
    
\end{example}

The next theorem shows that the marginal price auction system has higher allocative efficiency compared to the first-come-first-serve system (see Fig. \ref{FigAllocEff}).

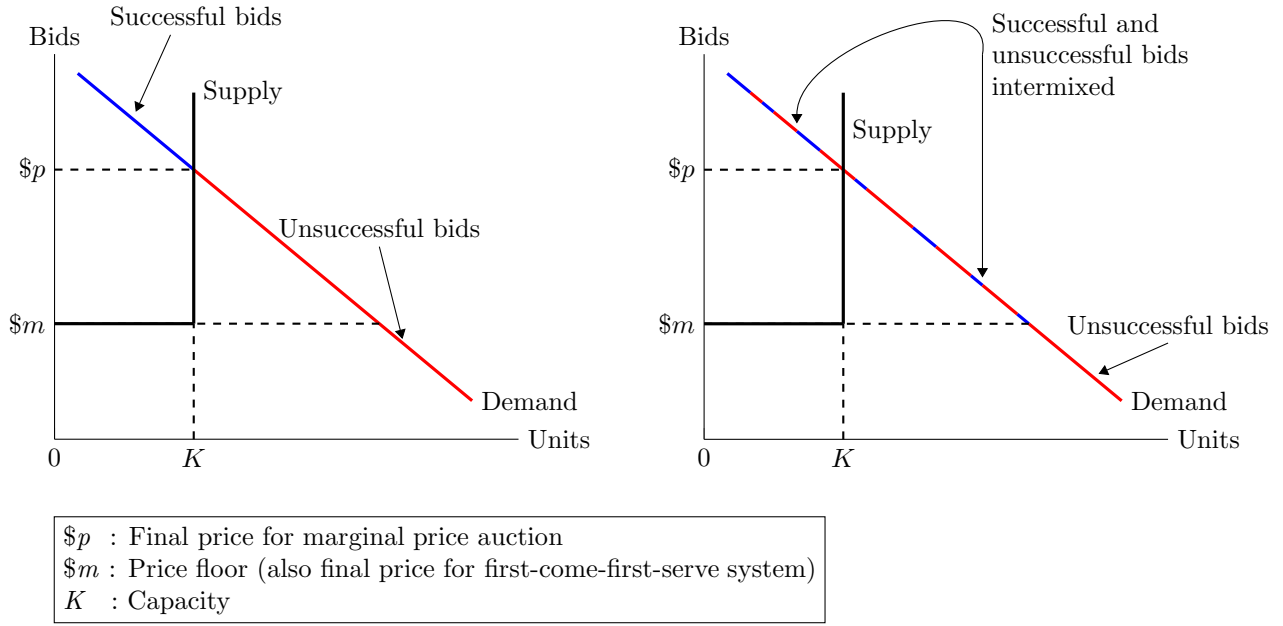
\begin{figure}[h]

\begin{center}
\hspace*{-2cm}
\begin{tikzpicture}
    \begin{axis}[
    scale = 0.9,
    xmin = 0, xmax = 10,
    ymin = 0, ymax = 10,
    axis lines* = left,
    xtick = {0}, ytick = \empty,
    clip = false,
    ]
    \addplot[color = blue, very thick] coordinates {(0.5, 9.5) (3, 7)};
    \addplot[color = red, very thick] coordinates {(3, 7) (9, 1)};
    \addplot[color = black, very thick] coordinates {(3, 3) (3, 9)};
    \addplot[color = black, very thick] coordinates {(0, 3) (3, 3)};
    
    \addplot[color = black, dashed, thick] coordinates {(0, 7) (3, 7)};
    \addplot[color = black, dashed, thick] coordinates {(0, 3) (7, 3)};
    \addplot[color = black, dashed, thick] coordinates {(3, 0) (3, 3)};
    
    \node [right] at (current axis.right of origin) {Units};
    \node [above] at (current axis.above origin) {Bids};
    \node [above] at (5, 5.2) {};
    \node [left] at (0, 7) {\$$p$};
    \node [left] at (0, 3) {\$$m$};
    \node [below] at (3, 0) {$K$};
    \node [right] at (9, 1) {Demand};
    \node [right] at (3, 9) {Supply};
    
    \node [right] at (1, 11) {Successful bids};
    \node [above] at (7, 5) {Unsuccessful bids};
    
    \draw[-Triangle] (2.5, 10.5) to (1.75, 8.5);
    \draw[-Triangle] (7, 5) to (7.5, 2.6);

    \node [below, draw, align = left] at (8.3, -2) {
    \$\textit{p } : Final price for marginal price auction \\
    \$\textit{m} : Price floor (also final price for first-come-first-serve system) \\
    \textit{K \space} : Capacity
    };
    
    \end{axis}

    \begin{axis}[
    scale = 0.9,
    xmin = 0, xmax = 10,
    ymin = 0, ymax = 10,
    axis lines* = left,
    xtick = {0}, ytick = \empty,
    clip = false,
    shift = {(axis cs: 14, 0)},
    ]
    \addplot[color = blue, very thick] coordinates {(0.5, 9.5) (1, 9)};
    \addplot[color = blue, very thick] coordinates {(1.25, 8.75) (1.5, 8.5)};
    \addplot[color = blue, very thick] coordinates {(2, 8) (2.5, 7.5)};
    \addplot[color = blue, very thick] coordinates {(3.25, 6.75) (3.5, 6.5)};
    \addplot[color = blue, very thick] coordinates {(4.5, 5.5) (5, 5)};
    \addplot[color = blue, very thick] coordinates {(5.75, 4.25) (6, 4)};
    \addplot[color = blue, very thick] coordinates {(6.75, 3.25) (7, 3)};

    \addplot[color = red, very thick] coordinates {(1, 9) (1.25, 8.75)};
    \addplot[color = red, very thick] coordinates {(1.5, 8.5) (2, 8)};
    \addplot[color = red, very thick] coordinates {(2.5, 7.5) (3.25, 6.75)};
    \addplot[color = red, very thick] coordinates {(3.5, 6.5) (4.5, 5.5)};
    \addplot[color = red, very thick] coordinates {(5, 5) (5.75, 4.25)};
    \addplot[color = red, very thick] coordinates {(6, 4) (6.75, 3.25)};

    \addplot[color = red, very thick] coordinates {(7, 3) (9, 1)};
    \addplot[color = black, very thick] coordinates {(3, 3) (3, 9)};
    \addplot[color = black, very thick] coordinates {(0, 3) (3, 3)};
    
    \addplot[color = black, dashed, thick] coordinates {(0, 7) (3, 7)};
    \addplot[color = black, dashed, thick] coordinates {(0, 3) (7, 3)};
    \addplot[color = black, dashed, thick] coordinates {(3, 0) (3, 3)};
    
    \node [right] at (current axis.right of origin) {Units};
    \node [above] at (current axis.above origin) {Bids};
    \node [above] at (5, 5.2) {};
    \node [left] at (0, 7) {\$$p$};
    \node [left] at (0, 3) {\$$m$};
    \node [below] at (3, 0) {$K$};
    \node [right] at (9, 1) {Demand};
    \node [right] at (3, 8) {Supply};
    
    \node [text width=3cm,right] at (6, 10) {Successful and \\unsuccessful bids intermixed};
    \node [above] at (10, 2.5) {Unsuccessful bids};
    
    \draw[-Triangle] (6, 10) to (6, 4.2);
    \draw[-Triangle] (10, 2.5) to (8.5, 1.7);
    \draw[-Triangle] (6, 10) to [out = 80, in = 90] (2, 8.2);
    \end{axis}
    
\end{tikzpicture}
\hspace*{-2cm}
\end{center}

    \caption{Ticket Allocation of a Marginal Price Auction System (L) and a First-Come-First-Serve System (R)}
    \label{FigAllocEff}
\end{figure}

\begin{theorem}
\label{Theorem3}
Assuming single-unit bidder demand and price floor, the sum of the valuations of the successful bidders in a marginal price auction system is greater than or equal to the sum of the valuations of successful bidders in a first-come-first-serve system.
\end{theorem}

\begin{proof}

    Let $N$ denote the number of bidders in the auction. Let the price floor be denoted by \$$m$ and $v_i$ denote bidder $i$'s valuation for one unit of the item. Let $\textbf{b}^{b_i\geq m}=\left(b_1, b_2, \ldots  b_i, \ldots b_k\right)$ denote the $k$-vector of submitted bids that exceed the price floor arranged in chronological order with $k \leq N$, and let $\textbf{c}^{b_i\geq m}=\left(c_1, c_2, \ldots  c_i, \ldots c_k\right)$ be the sorted $\textbf{b}^{b_i\geq m}$ vector in descending price order.

    The marginal price auction system allocates the units based on the leading entries of vector $\textbf{c}^{b_i\geq m}$ while the first-come-first-serve system allocates units based on the leading entries of vector $\textbf{b}^{b_i\geq m}$. Since $\textbf{c}^{b_i\geq m}$ is sorted based on descending price order, the sum of its leading entries is greater than or equal to the sum of the leading entries of $\textbf{b}^{b_i\geq m}$. From Theorem \ref{Theorem1}, we know that the optimal bidding strategy is $b_i=v_i$. Hence the sum of the bids is equal to the sum of the valuations. Thus, the sum of the valuations of the successful bidders in a marginal price auction system is greater than or equal to the sum of the valuations of successful bidders in a first-come-first-serve system.
    
\end{proof}

While the marginal price auction system does improve overall allocative efficiency, it nevertheless erodes consumer surplus and redistributes the surplus to the sellers (see Fig. \ref{FigSurplus}). To ensure that consumers still enjoy some benefits of switching to the marginal price ticketing system, a welfare transfer in the form of a rebate and/or excess collateral return mechanism may be implemented if the final settlement price greatly exceeds the original price floor ticket price (see Fig. \ref{FigReturn}). Non-monetary rebates (e.g. merchandise) may also be distributed if there is perceived value by the bidders. It is important to note that this rebate must be done after the auction has taken place, and should not occur frequently enough so as to change the expectations of the bidders. Changing expectations will result in deviations in the optimal strategy of bidders, and could render Theorem \ref{Theorem1} invalid.

\begin{figure}[h]

\begin{center}
\hspace*{-2cm}
\begin{tikzpicture}
    \begin{axis}[
    scale = 0.9,
    xmin = 0, xmax = 10,
    ymin = 0, ymax = 10,
    axis lines* = left,
    xtick = {0}, ytick = \empty,
    clip = false,
    ]
    \addplot[color = blue, very thick] coordinates {(0.5, 9.5) (3, 7)};
    \addplot[color = red, very thick] coordinates {(3, 7) (9, 1)};
    \addplot[color = black, very thick] coordinates {(3, 3) (3, 9)};
    \addplot[color = black, very thick] coordinates {(0, 3) (3, 3)};
    
    \addplot[color = black, dashed, thick] coordinates {(0, 7) (3, 7)};
    \addplot[color = black, dashed, thick] coordinates {(0, 3) (7, 3)};
    \addplot[color = black, dashed, thick] coordinates {(3, 0) (3, 3)};
    
    \node [right] at (current axis.right of origin) {Units};
    \node [above] at (current axis.above origin) {Bids};
    \node [above] at (5, 5.2) {};
    \node [left] at (0, 7) {\$$p$};
    \node [left] at (0, 3) {\$$m$};
    \node [below] at (3, 0) {$K$};
    \node [right] at (9, 1) {Demand};
    \node [right] at (3, 9) {Supply};
    
    \node [right] at (1, 11) {Successful bids};
    \node [above] at (7, 5) {Unsuccessful bids};
    
    \draw[-Triangle] (2.5, 10.5) to (1.75, 8.5);
    \draw[-Triangle] (7, 5) to (7.5, 2.6);

    \node [below, draw, align = left] at (8.3, -2) {
    \$\textit{p } : Final price for marginal price auction \\
    \$\textit{m} : Price floor (also final price for first-come-first-serve system) \\
    \textit{K \space} : Capacity \\
    \fcolorbox{black}{teal!10}{\makebox(0.13cm,0.13cm){\space}} : Consumer surplus \\
    \fcolorbox{black}{violet!10}{\makebox(0.13cm,0.13cm){\space}} : Seller surplus
    };

    \fill[teal, opacity = 0.1] (0.5, 7) -- (3, 7) -- (0.5, 9.5);
    \fill[violet, opacity = 0.1] (0.5, 0) -- (3, 0) -- (3, 7) -- (0.5, 7);

    \end{axis}

    \begin{axis}[
    scale = 0.9,
    xmin = 0, xmax = 10,
    ymin = 0, ymax = 10,
    axis lines* = left,
    xtick = {0}, ytick = \empty,
    clip = false,
    shift = {(axis cs: 14, 0)},
    ]
    \addplot[color = blue, very thick] coordinates {(0.5, 9.5) (1, 9)};
    \addplot[color = blue, very thick] coordinates {(1.25, 8.75) (1.5, 8.5)};
    \addplot[color = blue, very thick] coordinates {(2, 8) (2.5, 7.5)};
    \addplot[color = blue, very thick] coordinates {(3.25, 6.75) (3.5, 6.5)};
    \addplot[color = blue, very thick] coordinates {(4.5, 5.5) (5, 5)};
    \addplot[color = blue, very thick] coordinates {(5.75, 4.25) (6, 4)};
    \addplot[color = blue, very thick] coordinates {(6.75, 3.25) (7, 3)};

    \addplot[color = red, very thick] coordinates {(1, 9) (1.25, 8.75)};
    \addplot[color = red, very thick] coordinates {(1.5, 8.5) (2, 8)};
    \addplot[color = red, very thick] coordinates {(2.5, 7.5) (3.25, 6.75)};
    \addplot[color = red, very thick] coordinates {(3.5, 6.5) (4.5, 5.5)};
    \addplot[color = red, very thick] coordinates {(5, 5) (5.75, 4.25)};
    \addplot[color = red, very thick] coordinates {(6, 4) (6.75, 3.25)};

    \addplot[color = red, very thick] coordinates {(7, 3) (9, 1)};
    \addplot[color = black, very thick] coordinates {(3, 3) (3, 9)};
    \addplot[color = black, very thick] coordinates {(0, 3) (3, 3)};
    
    \addplot[color = black, dashed, thick] coordinates {(0, 7) (3, 7)};
    \addplot[color = black, dashed, thick] coordinates {(0, 3) (7, 3)};
    \addplot[color = black, dashed, thick] coordinates {(3, 0) (3, 3)};
    
    \node [right] at (current axis.right of origin) {Units};
    \node [above] at (current axis.above origin) {Bids};
    \node [above] at (5, 5.2) {};
    \node [left] at (0, 7) {\$$p$};
    \node [left] at (0, 3) {\$$m$};
    \node [below] at (3, 0) {$K$};
    \node [right] at (9, 1) {Demand};
    \node [right] at (3, 8) {Supply};
    
    \node [text width=3cm,right] at (6, 10) {Successful and \\unsuccessful bids intermixed};
    \node [above] at (10, 2.5) {Unsuccessful bids};
    
    \draw[-Triangle] (6, 10) to (6, 4.2);
    \draw[-Triangle] (10, 2.5) to (8.5, 1.7);
    \draw[-Triangle] (6, 10) to [out = 80, in = 90] (2, 8.2);

    \fill[teal, opacity = 0.1] (0.5, 9.5) -- (1, 9) -- (0.5, 9);
    \fill[teal, opacity = 0.1] (0.5, 3) -- (1, 3) -- (1, 9) -- (0.5, 9);

    \fill[teal, opacity = 0.1] (1.25, 8.75) -- (1.5, 8.5) -- (1.25, 8.5);
    \fill[teal, opacity = 0.1] (1.25, 3) -- (1.5, 3) -- (1.5, 8.5) -- (1.25, 8.5);

    \fill[teal, opacity = 0.1] (2, 8) -- (2.5, 7.5) -- (2, 7.5);
    \fill[teal, opacity = 0.1] (2, 3) -- (2.5, 3) -- (2.5, 7.5) -- (2, 7.5);

    \fill[teal, opacity = 0.1] (3.25, 6.75) -- (3.5, 6.5) -- (3.25, 6.5);
    \fill[teal, opacity = 0.1] (3.25, 3) -- (3.5, 3) -- (3.5, 6.5) -- (3.25, 6.5);

    \fill[teal, opacity = 0.1] (4.5, 5.5) -- (5, 5) -- (4.5, 5);
    \fill[teal, opacity = 0.1] (4.5, 3) -- (5, 3) -- (5, 5) -- (4.5, 5);

    \fill[teal, opacity = 0.1] (5.75, 4.25) -- (6, 4) -- (5.75, 4);
    \fill[teal, opacity = 0.1] (5.75, 3) -- (6, 3) -- (6, 4) -- (5.75, 4);

    \fill[teal, opacity = 0.1] (6.75, 3.25) -- (7, 3) -- (6.75, 3);

    \fill[violet, opacity = 0.1] (0.5, 0) -- (1, 0) -- (1, 3) -- (0.5, 3);
    \fill[violet, opacity = 0.1] (1.25, 0) -- (1.5, 0) -- (1.5, 3) -- (1.25, 3);
    \fill[violet, opacity = 0.1] (2, 0) -- (2.5, 0) -- (2.5, 3) -- (2, 3);
    \fill[violet, opacity = 0.1] (3.25, 0) -- (3.5, 0) -- (3.5, 3) -- (3.25, 3);
    \fill[violet, opacity = 0.1] (4.5, 0) -- (5, 0) -- (5, 3) -- (4.5, 3);
    \fill[violet, opacity = 0.1] (5.75, 0) -- (6, 0) -- (6, 3) -- (5.75, 3);
    \fill[violet, opacity = 0.1] (6.75, 0) -- (7, 0) -- (7, 3) -- (6.75, 3);
    \end{axis}
    
\end{tikzpicture}
\hspace*{-2cm}
\end{center}

    \caption{Consumer and Seller Surplus of a Marginal Price Auction System (L) and a First-Come-First-Serve System (R)}
    \label{FigSurplus}
\end{figure}

\begin{figure}[h]

\begin{center}
\begin{tikzpicture}
    \begin{axis}[
    scale = 1.2,
    xmin = 0, xmax = 10,
    ymin = 0, ymax = 10,
    axis lines* = left,
    xtick = {0}, ytick = \empty,
    clip = false,
    ]
    \addplot[color = blue, very thick] coordinates {(0.5, 9.5) (3, 7)};
    \addplot[color = red, very thick] coordinates {(3, 7) (9, 1)};
    \addplot[color = black, very thick] coordinates {(3, 3) (3, 9)};
    \addplot[color = black, very thick] coordinates {(0, 3) (3, 3)};
    
    \addplot[color = black, dashed, thick] coordinates {(0, 7) (3, 7)};
    \addplot[color = black, dashed, thick] coordinates {(0, 3) (7, 3)};
    \addplot[color = black, dashed, thick] coordinates {(3, 0) (3, 3)};
    
    \node [right] at (current axis.right of origin) {Units};
    \node [above] at (current axis.above origin) {Bids};
    \node [above] at (5, 5.2) {};
    \node [left] at (0, 7) {\$$p$};
    \node [left] at (0, 3) {\$$m$};
    \node [below] at (3, 0) {$K$};
    \node [right] at (9, 1) {Demand};
    \node [right] at (3, 9) {Supply};
    
    \node [right] at (1, 11) {Successful bids};
    \node [above] at (7, 5) {Unsuccessful bids};

    \draw[|-|] (-1, 7) to (-1, 4.5);
    \node [below, align = left] at (-2.3, 6) {Rebate};
    
    \draw[-Triangle] (2.5, 10.5) to (1.75, 8.5);
    \draw[-Triangle] (7, 5) to (7.5, 2.6);

    \node [below right, draw, align = left] at (9, 10.5) {
    \$\textit{p } : Final price \\
    \$\textit{m} : Price floor \\
    \textit{K \space} : Capacity \\
    \fcolorbox{black}{teal!10}{\makebox(0.13cm,0.13cm){\space}} : Consumer surplus \\
    \fcolorbox{black}{violet!10}{\makebox(0.13cm,0.13cm){\space}} : Seller surplus \\
    \fcolorbox{black}{orange!10}{\makebox(0.13cm,0.13cm){\space}} : Welfare transfer
    };

    \fill[teal, opacity = 0.1] (0.5, 7) -- (3, 7) -- (0.5, 9.5);
    \fill[violet, opacity = 0.1] (0.5, 0) -- (3, 0) -- (3, 4.5) -- (0.5, 4.5);
    \fill[orange, opacity = 0.1] (0.5, 4.5) -- (3, 4.5) -- (3, 7) -- (0.5, 7);
    \end{axis}
    
\end{tikzpicture}
\end{center}
    \caption{Welfare Transfer from Sellers to Consumers from Rebate}
    \label{FigReturn}
\end{figure}
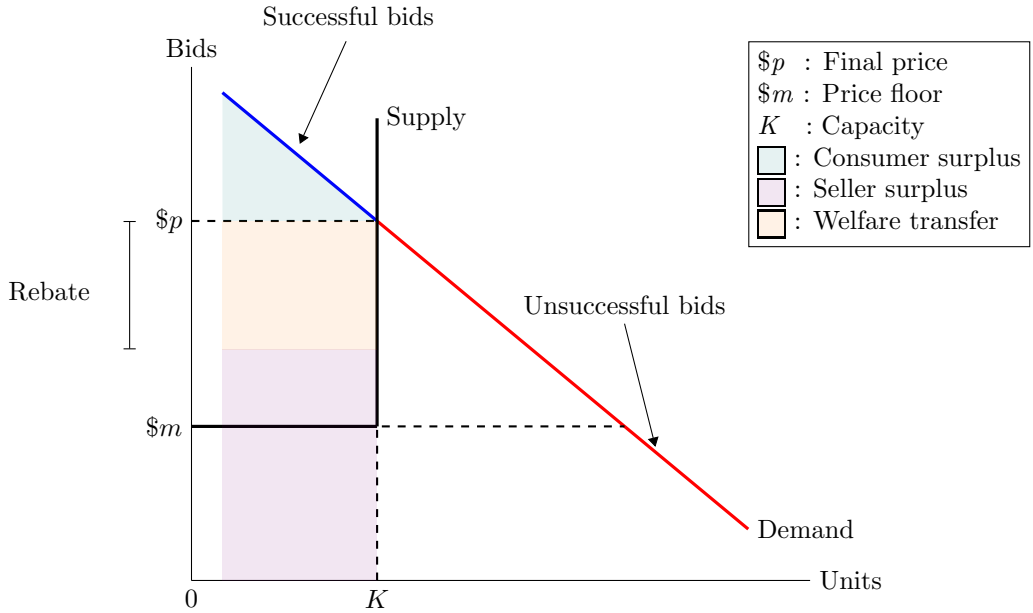

Finally, we prove that it is not optimal to buy from ticket scalpers in the case of time-invariant bidder valuations.

\begin{theorem}
    \label{Theorem4}
    If individual valuations are time-invariant, then it is not optimal for bidders to buy from ticket scalpers after an unsuccessful bid.
\end{theorem}

\begin{proof}
    Let $p$ be the final price set by the marginal price auction and let $v_i$ denote bidder $i$'s valuation for one ticket. If bidder $i$ is unsuccessful in the auction, by Theorem \ref{Theorem1}, the individual valuation is less than the final price ($v_i < p$). For economically rational ticket scalpers, the scalping price $\Tilde{p}$ is given by $\Tilde{p} \geq p$. Assuming individual valuations are time-invariant, we have $v_i < p \leq \Tilde{p}$. So bidder $i$'s valuation of the ticket is below the scalping price, and the bidder is better off not buying the ticket from the scalper.
\end{proof}

Theorem \ref{Theorem4} is particularly relevant for event ticketing purposes as it partially solves the ticket scalping problem. Although not necessarily a negative externality in the economics sense as ticket scalpers do serve a purpose to equilibrate limited supply with demand, it is nevertheless regarded as socially unacceptable and banned in most countries due to the erosion of consumer surplus (for the case of Australia, see Ref. \cite{TicketReselling2017}). The marginal price auction mechanism partially solves this as bidders with time-invariant valuations will refrain from purchasing tickets from scalpers. Therefore, individuals that could potentially buy from scalpers are restricted to the subset of bidders that have time-varying valuations, new bidders that did not participate in the original auction and/or bidders who may wish to obtain better seating for the event.

\section{Simulation}

We provide a simple simulation of the marginal price auction system summarized in Table \ref{TableSim}. The simulation assumes three scenarios covering small, medium, and large events with capacity $K=100$, $1000$, and $10000$ respectively. We also assume a price floor of $m=100$ with the number of bidders equaling $1.5 \times K$ and have valuations according to the normal distribution N($\mu=125$, $\sigma=25$) (see Ref. \cite{Koebert2023}). The emphasis here is not on the assumptions and such analysis is best left to the econometricians. Here we focus on key distinctive features of the marginal price ticket allocation system.

\begin{table}
\setlength\tabcolsep{0pt}
\begin{tabular*}{\textwidth}{@{\extracolsep{\fill}} ll *{4}{r} }
\toprule
  && \multicolumn{4}{c}{Summary Statistics} \\
\cmidrule{3-6} 
  && \multicolumn{3}{c}{Total} 
   & \multicolumn{1}{c}{Per Person} \\
\cmidrule{3-5} \cmidrule{6-6}
\mr{} & Min & Mean & Max & Mean \\
\midrule
$K=100$  
& Ticket Price 
    & 101.09 & 111.91 & 123.12 & 111.91\\
& Excess Economic Rents
    & 109.00 & 1191.17 & 2312.00 & 11.91\\
& Excess Valuation
    & 38.12 & 542.22 & 1078.50 & 5.42\\
& Excess Consumer Surplus
    & -1616.24 & -648.94 & 107.20 & -6.49\\
\midrule
$K=1000$
& Ticket Price 
    & 108.52 & 111.57 & 114.83 & 111.57\\
& Excess Economic Rents
    & 8520.0 & 11570.63 & 14830.00 & 11.57\\
& Excess Valuation
    & 4083.07 & 5542.01 & 7539.24 & 5.54\\
& Excess Consumer Surplus
    & -9169.78 & -6028.62 & -3448.89 & -6.03\\
\midrule
$K=10000$
& Ticket Price 
    & 110.47 & 111.53 & 112.58 & 111.53\\
& Excess Economic Rents
    & 104700.00 & 115342.40 & 125800.00 & 11.53\\
& Excess Valuation
    & 50447.43 & 55576.80 & 61630.25 & 5.56\\
& Excess Consumer Surplus
    & -68433.29 & -59765.60 & -50472.87 & -5.98\\
\bottomrule
\end{tabular*} 

\medskip
Note: The constants used in the simulation are: $m=100$ and $N=1.5 \times K$. The random sampling distribution used is the normal distribution $\mu=125$ and $\sigma=25$. The simulation was run 1000 times, its code is available at \url{https://github.com/williamfu54/Marginal-Price-Auction-Simulation}.

\caption{Summary Statistics of Marginal Price Auction Simulation}
\label{TableSim}
\end{table}

The simulation substantiates the proofs of Theorem \ref{Theorem2} and Theorem \ref{Theorem3}. We see an increase in both economic rents extracted and total bidder valuation from the marginal price auction system as compared to the first-come-first-serve system. However, we also see an erosion of consumer surplus due to the need to pay a higher ticket price. It may become socially unacceptable for the ticket price to be substantially above the price floor. In such cases, a rebate mechanism should be used to redistribute the surplus back to the consumers. Overall, the simulation shows that total allocative efficiency is increased by using the marginal price auction system.

\section{Conclusion}

Through this paper, we have analyzed a ticketing protocol based on the marginal price auction system. During the bidding window, bidders can submit bids for the tickets and post collateral. The protocol allocates the tickets to the highest bids and the ticket price is determined by the lowest winning bid. Tickets are then released to the successful bidders with a requirement to pay within a specified timeframe and collateral is given back to all bidders. We also proved that the mechanism allows for a more allocative efficient ticketing system. Additionally, more economic rents can be obtained by the event organizers and we also showed that it is not optimal for bidders to buy from ticket scalpers under time-invariant valuations. Finally, we provide a simple simulation to substantiate our proofs.

\newpage
\printbibliography

@incollection{Krishna2003,
title = {13 - Equilibrium and efficiency with private values},
editor = {Vijay Krishna},
booktitle = {Auction Theory},
publisher = {Academic Press},
address = {San Diego},
pages = {179-197},
year = {2003},
isbn = {978-0-12-426297-3},
doi = {https://doi.org/10.1016/B978-012426297-3.50039-4},
author = {Vijay Krishna}
}

@techreport{TicketReselling2017,
author = {{The Australian Government the Treasury}},
title = {Ticket Reselling in Australia},
year = {2017},
url = {https://oia.pmc.gov.au/published-impact-analyses-and-reports/ticket-reselling},
institution = {{Commonwealth of Australia}},
}

@incollection{Anthony2016,
title = {Chapter 2 - The Process View},
editor = {Richard John Anthony},
booktitle = {Systems Programming},
publisher = {Morgan Kaufmann},
address = {Boston},
pages = {21-106},
year = {2016},
isbn = {978-0-12-800729-7},
doi = {https://doi.org/10.1016/B978-0-12-800729-7.00002-9},
author = {Richard John Anthony},
}

@article{Vickrey1961,
author = {Vickrey, William},
title = {COUNTERSPECULATION, AUCTIONS, AND COMPETITIVE SEALED TENDERS},
journal = {The Journal of Finance},
volume = {16},
number = {1},
pages = {8-37},
doi = {https://doi.org/10.1111/j.1540-6261.1961.tb02789.x},
year = {1961}
}

@misc{Maiden2023,
author = {Samantha Maiden},
year = {2023},
title = {Truth about Taylor Swift tickets: Ticketek confirms there is actually no ‘queue’, defends system as ‘fair’},
url = {https://www.news.com.au/entertainment/music/truth-about-taylor-swift-tickets-ticketek-confirms-there-is-actually-no-queue-defends-system-as-fair/news-story/0ce6249be7a1df7b4ca028837ddee3a5},
note = {Accessed on 30 June 2023},
}

@article{McMurry2017,
    doi = {10.1371/journal.pbio.2001414},
    author = {McMurry, Julie A. AND Juty, Nick AND Blomberg, Niklas AND Burdett, Tony AND Conlin, Tom AND Conte, Nathalie AND Courtot, Mélanie AND Deck, John AND Dumontier, Michel AND Fellows, Donal K. AND Gonzalez-Beltran, Alejandra AND Gormanns, Philipp AND Grethe, Jeffrey AND Hastings, Janna AND Hériché, Jean-Karim AND Hermjakob, Henning AND Ison, Jon C. AND Jimenez, Rafael C. AND Jupp, Simon AND Kunze, John AND Laibe, Camille AND Le Novère, Nicolas AND Malone, James AND Martin, Maria Jesus AND McEntyre, Johanna R. AND Morris, Chris AND Muilu, Juha AND Müller, Wolfgang AND Rocca-Serra, Philippe AND Sansone, Susanna-Assunta AND Sariyar, Murat AND Snoep, Jacky L. AND Soiland-Reyes, Stian AND Stanford, Natalie J. AND Swainston, Neil AND Washington, Nicole AND Williams, Alan R. AND Wimalaratne, Sarala M. AND Winfree, Lilly M. AND Wolstencroft, Katherine AND Goble, Carole AND Mungall, Christopher J. AND Haendel, Melissa A. AND Parkinson, Helen},
    journal = {PLOS Biology},
    publisher = {Public Library of Science},
    title = {Identifiers for the 21st century: How to design, provision, and reuse persistent identifiers to maximize utility and impact of life science data},
    year = {2017},
    month = {06},
    volume = {15},
    pages = {1-18},
    number = {6},
}

@misc{Koebert2023,
author = {Josh Koebert},
year = {2023},
title = {How Much Should Music Lovers of Each Genre Budget for Concerts This Summer?},
url = {https://financebuzz.com/price-of-concert-tickets-by-genre},
note = {Accessed on 30 June 2023},
}

\end{document}